\documentclass[11pt,a4paper,reqno]{amsart}%
\usepackage{amsthm,amsmath,amsfonts,amssymb,amsxtra,appendix,bookmark,dsfont,bbm,bm, hyperref,graphicx}

\usepackage[margin=1.15in]{geometry}

\theoremstyle{plain}
\newtheorem{theorem}{Theorem}
\newtheorem{lemma}[theorem]{Lemma}
\newtheorem{corollary}[theorem]{Corollary}
\newtheorem{proposition}[theorem]{Proposition}

\theoremstyle{definition}

\usepackage[english]{babel}

\theoremstyle{remark}
\newtheorem{remark}[theorem]{Remark}

\numberwithin{equation}{section}

\newcommand\bq{\begin{eqnarray}}
\newcommand\eq{\end{eqnarray}}
\newcommand\nn{\nonumber}
\renewcommand{\epsilon}{\varepsilon}
\newcommand\cF {\mathcal{F}}

\newcommand\cD {\mathcal{D}}

\newcommand\R {\mathbb{R}}

\title[2D  Bose gas ground state energy from Bogoliubov functional]{Ground state energy of a dilute two-dimensional Bose gas from the Bogoliubov free  energy functional }

\author[S. Fournais]{S{\o}ren Fournais}
\address{Department of Mathematics, Aarhus University, Ny Munkegade 118, DK-8000 Aarhus C, Denmark}
\email{fournais@math.au.dk}

\author[M. Napi\'orkowski]{Marcin Napi\'orkowski}
\address{Department of Mathematics, LMU Munich, Theresienstrasse 39, 80333 Munich, Germany \& \newline
Department of Mathematical Methods in Physics, Faculty of Physics, University of Warsaw, Pasteura 5, 02-093 Warsaw, Poland}
\email{marcin.napiorkowski@fuw.edu.pl}

\author[R. Reuvers]{Robin Reuvers}
\address{DAMTP, Centre for Mathematical Sciences, University of Cambridge, Wilberforce Road, Cambridge CB3 0WA, United Kingdom} 
\email{r.reuvers@damtp.cam.ac.uk}

\author[J.~P. Solovej]{Jan Philip Solovej}
\address{QMATH, Department of Mathematical Sciences, University of Copenhagen, Universitetsparken 5, DK-2100 Copenhagen \O, Denmark} 
\email{solovej@math.ku.dk}

\begin{document}
\date{\today}

\begin{abstract}   
  We extend the analysis of the Bogoliubov free energy functional to two dimensions at very low temperatures.  For sufficiently weak interactions, we prove two term asymptotics for the ground state energy. 
\end{abstract}

\maketitle

\section{Introduction}
It has been a long standing problem to determine the asymptotics of the ground state energy density $e_0 (n,\rho)$ of a dilute Bose gas in  $n$ dimensions  in the thermodynamic limit.

In three spatial dimensions ($n=3$), the famous Lee--Huang--Yang formula \cite{LeeHuaYan-57} predicts that   
\begin{equation}\label{eq:LHY}
e_0 (3,\rho)=4\pi a \rho^2 \left(1 +\frac{128}{15 \sqrt{\pi}}(\rho a^3)^{1/2}+o(\sqrt{\rho a^3})\right) \,\,\, \text{as} \,\,\, \rho a^3 \to 0,
\end{equation}
where $a$ denotes the scattering length of the two-body interaction potential, $\rho$ is the density of the system and we have chosen units so that $\hbar/2m=1$. The diluteness condition
$\rho a^3 \ll 1 $ means that average inter-particle spacing $\rho^{-1/3}$ is much smaller than the effective range of the potential $a$. Note that this formula implies that the ground state energy depends on the two-body interaction only through the scattering length and can be interpreted as a universality result.

The first rigorous result in the direction of \eqref{eq:LHY} has been obtained by Dyson in \cite{Dyson-57} where a matching upper bound for the first term in the expansion has been obtained. It then took over 40 years until Lieb and Yngvason \cite{LieYng-98} proved a lower bound for the leading order term (see also \cite{Lee-09,Yin-10b,LeeYin-10} for further improvements and extensions). The correct upper bound for the second term in the expansion has been obtained by Yau and Yin \cite{YauYin-09} (see also \cite{ErdSchYau-08} for an intermediate result). A precise lower bound that captures the second order term has been proven only very recently in \cite{FouSol-19} (see (\cite{GiuSei-09,BriSol-17,BriFouSol-18} for partial results). 

One can ask the same question in two spatial dimensions ($n=2$). In this case it has been first derived by Schick \cite{Schick-71} and then proved by Lieb and Yngvason \cite{LieYng-01} that
\begin{equation}\label{eq:Schick}
e_0 (2,\rho)=\frac{4\pi \rho^2}{|\ln(\rho a^2)|}\left(1+o(1)\right) \,\,\, \text{as} \,\,\, \rho a^2 \to 0.
\end{equation}
While there are no rigorous results about the second order correction term, it has been predicted in many works (see, e.g., \cite{Andersen-02,PilBorCasGio-05,MorCas-09}) that it should be given by 
\begin{equation} \label{eq:secondorderconjecture}
-\frac{4\pi \rho^2 \ln|\ln(\rho a^2)| }{|\ln(\rho a^2)|^2}.
\end{equation}
An even more difficult task is to consider the Bose gas at positive temperatures. In principle, all thermodynamical properties of such a system are accessible via its free energy. There exist only very limited rigorous results about the free energy of a bosonic system starting from a many-body Schr\"odinger Hamiltonian. In fact, in the dilute regime, the homogeneous gas in three dimensions has been treated by Seiringer \cite{Sei-08} and Yin \cite{Yin-10a} (see also \cite{DeuSeiYng-18,DeuSei-19} for recent developments for the trapped and Gross--Pitaevskii cases).
The lower \cite{Sei-08} and upper bound \cite{Yin-10a} prove the free energy asymptotics 
\begin{equation} \label{eq:3D_free_energy_SeiYin}
F(T,\rho)=F_0 (T,\rho)+4\pi a \left( 2\rho^2 - [\rho-\rho_{\rm{fc}} ]_+ ^2\right) + o(a\rho^2) \,\,\, \text{as} \,\,\, \rho a^3 \to 0.
\end{equation} 
Here $F_0 (T,\rho)$ is the free energy of the ideal Bose gas given by
\bq \label{def:idealfreeenergy}
F_0 (T,\rho)=\sup_{\mu \leq 0} \left\{ \mu \rho +\frac{T}{(2\pi)^3} \int_{\mathbb{R}^3} \ln \left(1- e^{-T^{-1}(p^2-\mu)} \right) dp  \right\}.
\eq
The free critical density for Bose--Einstein condensation (BEC) in the ideal gas equals 
\bq \label{def:freecritdensity}
\rho_{\rm{fc}} (T)= \frac{1}{(2\pi)^3} \int_{\mathbb{R}^3} \left(e^{p^2/T} -1 \right)^{-1} dp. \eq
Finally, $[\cdot]_+ = \max \{\cdot,0\}$  denotes the positive part. Very recently, the free energy of a homogeneous two-dimensional Bose gas has been rigorously analyzed by Deuchert, Mayer and Seiringer \cite{DeuMaySei-18}.

While the impressive results mentioned above improve our understanding of Bose gases at positive temperatures significantly, the proof of Bose--Einstein Condensation in a homogeneous system in the thermodynamic limit is still lacking.  The analysis of certain properties of Bose systems at positive temperatures relies therefore only on effective theories which - while still non-trivial - are easier to handle than the full many-body problem. 

One of such effective theories is given by the so-called Bogoliubov free energy functional. This functional, first introduced by Critchley and Solomon  \cite{CriSol-76}, models a homogeneous Bose gas at positive temperatures in the thermodynamic limit. It is obtained by evaluating the expectation value of $H-TS-\mu \mathcal{N}$ in a quasi-free state (here $H$ is the Hamiltonian, $S$ is the von Neumann entropy, $\mathcal{N}$ is the particle number operator and $\mu$ us the chemical potential) and then varying over all quasi-free states. 

 Recently, several interesting properties of a system of interacting bosons have been derived within this approximation. In particular, it has been shown in \cite{NapReuSol-15a} that the model exhibits a BEC phase transition and, in the dilute limit, the critical temperature of such phase transition has been derived \cite{NapReuSol-15b}. The Lee--Huang--Yang formula \eqref{eq:LHY} and the free energy expansion \eqref{eq:3D_free_energy_SeiYin} have been also almost (this will be made precise later) recovered. In \cite{NapReuSol-17} it has been shown that the functional has also interesting features in two dimensions at criticality. 

The goal of this paper is to extend the analysis in \cite{NapReuSol-17} to include a derivation of the ground state energy formula that includes both \eqref{eq:Schick} and \eqref{eq:secondorderconjecture}. A precise statement of the theorem will be given in Section \ref{sec:functionalandresult} where the functional will be also introduced. In Section \ref{sec:generalresults} we will recall its most important properties, summarizing the results of \cite{NapReuSol-15a,NapReuSol-15b,NapReuSol-17}. In Section \ref{sec:simplifiedfunctional} we will provide an outline of the proof together with the most important preliminary steps. The proof of the main result will be a  given in Section \ref{sec:2DfreeenergylowT}.
\medskip
\medskip

\section{The model and main result}\label{sec:functionalandresult}

The model we want to analyze is defined by the \textit{Bogoliubov free energy functional} $\cF$ given by
\bq
\label{def:grandcanfreeenergyfunctional}
\begin{aligned}
\mathcal{F}(\gamma,\alpha,\rho_{0})&= (2\pi)^{-n}\int_{\mathbb{R}^n} p^{2}\gamma(p)dp-\mu\rho-TS(\gamma,\alpha)+\frac{\widehat{V}(0)}{2}\rho^{2} \\
&+\frac{1}{2}(2\pi)^{-2n}\iint_{\mathbb{R}^n\times \mathbb{R}^n}\widehat{V}(p-q)\left(\alpha(p)\alpha(q)+\gamma(p)\gamma(q)\right)dpdq \\ &+\rho_{0}(2\pi)^{-n}\int_{\mathbb{R}^n}\widehat{V}(p)\left(\gamma(p)+\alpha(p)\right)dp,
\end{aligned}
\eq
which is the free energy expectation value in a quasi-free state. Here, $\rho$ denotes the density of the system and 
\bq
\rho=\rho_0+(2\pi)^{-n}\int_{\mathbb{R}^n}\gamma(p)dp=:\rho_0+\rho_\gamma. \nn
\eq
The entropy $S(\gamma,\alpha)$ is 
\begin{equation}
\label{entro}
\begin{aligned}
&S(\gamma,\alpha)= \quad(2\pi)^{-n}\int_{\mathbb{R}^n}s(\gamma(p),\alpha(p))dp\quad=\quad(2\pi)^{-n}\int_{\mathbb{R}^n}s(\beta(p))dp \\
 &=(2\pi)^{-n}\int_{\mathbb{R}^n}\left[\left(\beta(p)+\frac{1}{2}\right)\ln\left(\beta(p)+\frac{1}{2}\right)-\left(\beta(p)-\frac{1}{2}\right)\ln\left(\beta(p)-\frac{1}{2}\right)\right]dp,
\end{aligned}
\end{equation}
where
\bq
\beta(p):=\sqrt{\left(\frac{1}{2}+\gamma(p)\right)^{2}-\alpha(p)^{2}}.\label{def:beta}
\eq
The functional is defined on the domain $\mathcal{D}$ given by
\bq
\mathcal{D}=\{(\gamma,\alpha,\rho_{0})|\gamma \in L^{1}((1+p^{2})dp),\gamma(p)\geq0, \alpha(p)^{2}\leq\gamma(p)(1+\gamma(p)), \rho_{0}\geq 0\}. \nn
\eq

This set-up describes the grand canonical free energy of a homogeneous Bose gas at temperature $T\geq0$ and chemical potential $\mu\in \mathbb{R}$ in the thermodynamic limit. The grand-canonical free energy is given by
\begin{align}
\label{gcmin}
F(T,\mu)=\inf_{(\gamma,\alpha,\rho_{0})\in \mathcal{D}}\mathcal{F}(\gamma,\alpha,\rho_{0}).
\end{align}
Throughout the paper we will assume that the interaction potential is described through a positive, radial, smooth and compactly supported function $V(x)$ whose Fourier transform 
$$\widehat{V}(p)=\int_{\R^n}e^{-ipx}V(x)dx$$
is also positive.

Let us very briefly mention how \eqref{def:grandcanfreeenergyfunctional} is obtained. We start from the Hamiltonian for a gas of $N$ bosons with a repulsive pair interaction $V$ in a $n$-dimensional box $\left[-l/2,l/2\right]^n$ and periodic boundary conditions, where $n=2,3$. In units $\hbar=2m=k_B=1$, 
\begin{equation}
\label{HN1}
H_N=\sum_{1\leq i\leq N}-\Delta_i+\sum_{1\leq i<j\leq N}V_{ij},
\end{equation}
with second-quantized form in momentum space
\begin{equation}
\label{HN}
H=\sum_p p^2 a^\dagger_p a_p+\frac{1}{2l^n}\sum_{p,q,k} \widehat{V}(k) a_{p+k}^\dagger a_{q-k}^\dagger a_q a_p.
\end{equation}
The canonical Gibbs state at temperature $T$ and particle density $\rho=N/l^n$ can be found by minimizing
\begin{equation}
\label{tomin2}
\inf_{\omega}\big[\langle H_N\rangle_\omega-TS(\omega)\big],
\end{equation}
where $\omega$ is an $N$-boson state and $S$ is the von Neumann entropy.

The grand canonical Gibbs state at temperature $T$ and chemical potential $\mu$ is the minimizer of 
\begin{equation}
\label{tomin}
\inf_{\omega}\big[\langle H-\mu\mathcal{N}\rangle_\omega-TS(\omega)\big],
\end{equation}
where $\omega$ is now a state on the bosonic Fock space, $\mathcal{N}$ is the particle number operator and the infimum itself is the free energy. 

The states $\omega$ which we want to consider in \eqref{tomin} are quasi-free states with an added condensate which we expect to have zero momentum. To implement the possibility of the occurrence of a condensate, one replaces $a_0\rightarrow a_0+\sqrt{l^n\rho_0}$ in the Hamiltonian. This is the counterpart of the famous $c$-number substitution of Bogoliubov  \cite{Bogoliubov-47b} which has been rigorously justified in \cite{LieSeiYng-05}. One then evaluates the expectation value of the resulting Hamiltonian among quasi-free states only. Those states satisfy Wick's rule and thus 
\begin{equation}
\langle a_{p+k}^\dagger a_{q-k}^\dagger a_q a_p\rangle =\langle a^\dagger_{p+k}a^\dagger_{q-k}\rangle\langle a_q a_p\rangle+\langle a^\dagger_{p+k}a_q\rangle\langle a^\dagger_{q-k}a_p\rangle+\langle a^\dagger_{p+k}a_p\rangle\langle a^\dagger_{q-k}a_q\rangle.
\end{equation}
Assuming translation invariance and $\langle a_p a_{-p}\rangle=\langle a^\dagger_{-p} a^\dagger_p\rangle$, the two (real-valued) functions \mbox{$\gamma(p):=\langle a^\dagger_p a_p\rangle\geq0$} and \mbox{$\alpha(p):=\langle a_p a_{-p}\rangle$}, together with the number $\rho_0$, now fully determine the expectation value in the Hamiltonian part of \eqref{tomin}. Taking the thermodynamic limit $l\to\infty$, yields the energy part of the functional.
One still has to determine the formula of the entropy of quasi-free states. For that we refer to the appendix of \cite{NapReuSol-15a}.

The derivation provides us with the following interpretation of the variables of the functional. The function $\gamma\in L^{1}((1+p^{2})dp)$ describes the momentum
distribution of the particles in the system. Since the total density
equals $\rho=\rho_0+(2\pi)^{-n}\int_{\mathbb{R}^n}\gamma(p)dp$, it
follows that a non-negative $\rho_0$ can be seen as the macroscopic
occupation of the state of momentum zero and is therefore interpreted as
the density of the Bose--Einstein condensate fraction.

Finally, the function $\alpha(p)$ describes pairing in the system and
its non-vanishing value can therefore be interpreted as the presence of
off-diagonal long-range order (ODLRO) and the macroscopic coherence
related to superfluidity.

It might be not immediately visible why the functional can be associated with Bogoliubov. The reason is the following.  Recall that Bogoliubov's approach \cite{Bogoliubov-47b} relies on two main assumptions: the $c$-number substitution and the truncation of the resulting Hamiltonian to a quadratic one. Now, ground and Gibbs states of such quadratic Hamiltonians are quasi-free states; exactly the states considered in our minimization problem. This is why we name the functional the way we do. Bogoliubov theory was extremely successful, therefore the hope that the functional will have interesting properties as well.  

So far we defined the model in the grand canonical ensamble. The diluteness condition $\rho^{1/2} a\ll 1$, under which we want to analyze the two-dimensional model, requires the notion of the density. To this end we formulate the canonical version of the functional. It is given by 
\bq
\label{def:canonicalfreeenergyfunctional}
\begin{aligned}
  \mathcal{F}^{\rm{can}}(\gamma,\alpha,\rho_{0})&= (2\pi)^{-n}\int_{\mathbb{R}^n} p^{2}\gamma(p)dp-TS(\gamma,\alpha)+\frac12\widehat{V}(0)\rho^{2}\\ &+\rho_{0}(2\pi)^{-n}\int_{\mathbb{R}^n}\widehat{V}(p)\left(\gamma(p)+\alpha(p)\right)dp \\
  &+\frac{1}{2}(2\pi)^{-2n}\iint_{\mathbb{R}^n\times
    \mathbb{R}^n}\widehat{V}(p-q)\left(\alpha(p)\alpha(q)+\gamma(p)\gamma(q)\right)dpdq,
\end{aligned}
\eq
with $\rho_0=\rho-\rho_\gamma$. The canonical minimization problem is
\begin{equation}
 \label{def:canonicalminimization}
\begin{aligned}
F^{\rm{can}}(T,\rho)&=\inf_{\substack{({\gamma},{\alpha},\rho_{0}=\rho-\rho_\gamma)\in\cD\\
    }}\cF^{\rm{can}}(\gamma,\alpha,\rho_0)=\inf_{0\leq\rho_0\leq\rho}f(\rho-\rho_0,\rho_0),
\end{aligned}
\end{equation}
where
\begin{equation}
\label{funcf}
f(\lambda,\rho_0)=\inf_{\substack{({\gamma},{\alpha})\in\cD'\\
      \int\gamma=\lambda
    }}\cF^{\rm{can}}(\gamma,\alpha,\rho_0)
\end{equation}
and 
\[
\mathcal{D}'=\{(\gamma,\alpha)\ |\ \gamma\in L^1((1+p^2)dp),\ \gamma(p)\geq0,\ \alpha(p)^2\leq \gamma(p)(\gamma(p)+1)\}.
\]
Strictly
speaking, this is not really a canonical formulation: it is only the
expectation value of the number of particles that we fix. We will
nevertheless describe this energy as canonical. The function
$F(T,\mu)$ as a function of $\mu$ is the Legendre transform of the
function $F^{\rm{can}}(T,\rho)$ as a function of $\rho$.

We are now ready to state the main result of this paper.
\begin{theorem}\label{thm:mainresult}
Consider $n=2$ and the dilute limit $\rho a^2\ll 1$. Let  $b:=|\ln(\rho a^2)|^{-1} \ll 1$ and let $F^{\rm{can}}(T,\rho)$ be as defined in \eqref{def:canonicalminimization}. Assume $\hat{V}(0)=\nu b$ for some positive parameter $\nu =O(1)$. Then
$$F^{\rm{can}}(0,\rho)=4\pi \rho^2 b + 4\pi \rho^2 b^2 \ln b +\left(\inf_{d\geq 0} C_\nu (d)\right)\rho^2 b^2 +  o \big(\rho^2 b^2 \big)$$
where
\begin{equation}\label{eq:Cvd}
\begin{aligned}
C_\nu (d)&=\left(1-\frac{1}{8\pi}\left(\sqrt{d(d+16\pi)}-d\right)\right)\left(2\nu -16\pi -d\right)+\frac{d^2}{16\pi}+2\pi+ d-4 \pi \left(\ln 8 - 2\Gamma\right) \\ &- \frac{1}{16\pi} d \sqrt{d (d+16 \pi )}-\frac12  \sqrt{d (d+16 \pi )}  +4 \pi  \ln \left(d+\sqrt{d (d+16 \pi )}+8 \pi \right).
\end{aligned}
\end{equation}
\end{theorem}
\begin{remark}
In the limit when $\nu \to 8\pi$ (c.f. Theorem \ref{thm:Crit2D}) the minimization over $d$ can be carried out explicitly. It is easy to see that in that case the infimum is attained for $d=0$ which corresponds to $C_{8\pi}(0)=2\pi(1+4\Gamma+2\ln \pi)$, where $\Gamma$ is the Euler–-Mascheroni constant.
\end{remark}

\section{Known results about the functional}\label{sec:generalresults}
\subsection{General properties of the functional.} The physical information about the system under consideration is encoded in the structure of the minimizers at given $(T,\mu)$ (or $(T,\rho)$). This means that the fundamental question that one needs to ask first is the one about the existence of minimizers of the functional. The positive answer to that question has been given in

\begin{theorem}[Theorems 2.1. - 2.4. in \cite{NapReuSol-15a}]\label{thm:well-posedness}
 Then the grand-canonical  (canonical) minimization problem has a minimizer for any $(T,\mu)$ (or $(T,\rho)$).
\end{theorem} 

Let us stress that for simplicity of the presentation, the assumptions on $V$ made in the previous Section are stronger than those in the original results. This also allows to treat the $T=0$ and $T>0$ cases together, while, in fact, the assumptions needed are not the same in those cases.

Knowing that minimizers of the functional exist, one can ask the question about their structure. It turns out that if the model exhibits a phase transition, then it does not distinguish between BEC and superfluidity as follows from

\begin{theorem}[Theorem 2.5. in \cite{NapReuSol-15a}]
\label{thm:BECvsSF}
  Let
  $(\gamma,\alpha,\rho_0)$ be a minimizing triple for
 either \eqref{def:grandcanfreeenergyfunctional} or \eqref{def:canonicalfreeenergyfunctional}. Then \bq \rho_0 =0
  \Longleftrightarrow \alpha \equiv 0. \nn \eq
\end{theorem}

Thus, there can only be one kind of phase transition, and the next results show that it indeed exists:

\begin{theorem}[Theorem 2.6. \cite{NapReuSol-15a}]\label{thm:phasetrangrandcan}
  Given $\mu>0$.  Then there exist temperatures $0<T_1<T_2$ such that
  a minimizing triple $(\gamma,\alpha,\rho_0)$ of
  \eqref{gcmin} satisfies
\begin{enumerate}
\item $\rho_0=0$ for $T\geq T_2$;
\item $\rho_0>0$ for $0\leq T \leq T_1$.
\end{enumerate}
\end{theorem} 

\begin{theorem}[Theorem 2.7. in \cite{NapReuSol-15a}]\label{thm:phasetrancan}
For fixed  $\rho>0$ there exist temperatures $0<T_3<T_4$ such that a minimizing triple
$(\gamma,\alpha,\rho_0)$ of \eqref{def:canonicalminimization} satisfies 
\begin{enumerate}
\item $\rho_0=0$ for $T\geq T_4$;
\item $\rho_0>0$ for $0\leq T \leq T_3$.
\end{enumerate}
\end{theorem} 

Knowing there exists a phase transition, one can ask for example what is the critical temperature. While there is no answer to that question in full generality, a relevant analysis can be made in the dilute limit. 

\subsection{Known results in the dilute limit in three dimensions.}
Since the diluteness conditions involves the scattering length $a$, let us recall its definition. In three dimensions it is given by    
\bq \label{def:scattlength3D}
4\pi a:=\int \Delta \varphi=\frac12 \int V\varphi, 
\eq
where $\varphi$ satisfies
\[
-\Delta \varphi+\frac12 V\varphi=0
\]
in the sense of distributions with $\varphi(x)\to1$ as $|x|\to\infty$. The quantity $8\pi a$ is often replaced by $\int V=\widehat{V}(0)$, which is its first-order Born approximation. In fact, $\widehat{V}(0)> 8\pi a$ (see \cite[Appendix C]{LieSeiSolYng-05} for more details). Below this discrepancy will be  quantified by the parameter $\nu=\widehat{V}(0)/a$, so that $\nu>8\pi$. The limit $\nu\to8\pi$, that is, a sequence of potentials such that $\widehat{V}(0)$ tends to $8\pi a$, is of special interest (cf. comment after Theorem \ref{thm:canfreeenexp}).\\

For the following three results in three dimensions, one has to assume that the gas is dilute
\begin{eqnarray}\label{dillim}
\rho^{1/3}a\ll 1,
\end{eqnarray}
and that there exists constant $C$ such that
\begin{equation}
\label{assumptionsV}
\int\widehat{V}\leq Ca^{-2}\hspace{1cm}\text{and}\hspace{1cm}\|\partial^n\widehat{V}\|_{\infty}\leq C a^{n+1}\text{\ for\ } 0\leq n\leq3,
\end{equation}
where $\partial^n$ is shorthand for all $n$-th order partial derivatives. The following theorems, proven in \cite{NapReuSol-15b}, contain information about the critical temperature of the phase transition in the dilute limit. 

\begin{theorem}[Theorem 8 in \cite{NapReuSol-15b}]
\label{thm:cancrittemp}
Let $(\gamma,\alpha,\rho_0)$ be a minimizing triple of \eqref{def:canonicalminimization} at temperature $T$ and density $\rho$. There is a monotone increasing function $h_1:(8\pi,\infty)\to\mathbb{R}$ with $h_1(\nu)\geq\lim_{\nu\to8\pi}h_1(\nu)=1.49$ such that
\begin{enumerate}
\item $\rho_0\neq 0$ if $T<T_{\rm fc}\left(1+h_1(\nu)\rho^{1/3}a+o(\rho^{1/3}a)\right)$
\item $\rho_0= 0$ if $T>T_{\rm fc}\left(1+h_1(\nu)\rho^{1/3}a+o(\rho^{1/3}a)\right)$,
\end{enumerate}
where $T_{\rm{fc}}=c_0\rho^{2/3}$ with $c_0=4\pi\xi(3/2)^{-2/3}$ is the critical temperature of the free Bose gas.
\end{theorem}

\begin{theorem}[Theorem 9 in \cite{NapReuSol-15b}]
\label{thm:grandcancrittemp}
Let $(\gamma,\alpha,\rho_0)$ be a minimizing triple of \eqref{gcmin} at temperature $T$ and chemical potential $\mu$. There is a function $h_2:(8\pi,\infty)\to\mathbb{R}$ with $\lim_{\nu\to8\pi}h_2(\nu)=0.44$ such that
\begin{enumerate}
\item $\rho_0\neq 0$ if $T<\left(\frac{\sqrt{\pi}}{2\zeta(3/2)}\frac{8\pi}{\nu}\right)^{2/3}\left(\frac{\mu}{a}\right)^{2/3}+h_2(\nu)\mu+o(\mu)$
\item $\rho_0= 0$ if $T>\left(\frac{\sqrt{\pi}}{2\zeta(3/2)}\frac{8\pi}{\nu}\right)^{2/3}\left(\frac{\mu}{a}\right)^{2/3}+h_2(\nu)\mu+o(\mu)$.
\end{enumerate}
\end{theorem}

The other main result of \cite{NapReuSol-15b} provides an expansion of the canonical free energy \eqref{def:canonicalminimization} in the dilute limit. For simplicity, we will state the results not in full generality but rather in the most relevant regions.

\begin{theorem}[Theorems 10, 11 and Corollary 12 in \cite{NapReuSol-15b}]
\label{thm:canfreeenexp}
Assume that $T$ and $\rho$ satisfy the conditions $\eqref{dillim}$ and $T\leq D\rho^{2/3}$ with $D>1$ fixed. Then 
\begin{enumerate}
\item For $T>T_{\rm fc}\left(1+h_1(\nu)\rho^{1/3}a+o(\rho^{1/3}a)\right)$, the free energy is 
$$F^{\rm{can}}(T,\rho)=F_0(T,\rho)+\widehat{V}(0)\rho^2+O((\rho a)^{5/2}),$$
and we have $\rho_\gamma=\rho$, $\rho_0=0$ for the minimizer. Here $F_0(T,\rho)$ is the free energy of the non-interacting gas.
\item  For $\rho a/T\ll 1$ and $T<T_{\rm fc}\left(1+h_1(\nu)\rho^{1/3}a+o(\rho^{1/3}a)\right)$, the canonical free energy is given by
\[
F^{\rm{can}}(T,\rho)=F_0(T,\rho)+4\pi a \rho^2+(\nu-4\pi)a\rho_{\rm fc}(2\rho-\rho_{\rm fc})+O(T(\rho a)^{3/2}).
\]

 \item  For $\rho a/T\gg 1$ and $T<T_{\rm fc}\left(1+h_1(\nu)\rho^{1/3}a+o(\rho^{1/3}a)\right)$, the canonical free energy can be described in terms of a function $g:(8\pi,\infty)\to\mathbb{R}$ as
\[
F^{\rm{can}}(T,\rho)=4\pi a \rho^2+g(\nu)(\rho a)^{5/2}+o\left((\rho a\right)^{5/2}),
\]
with $g(\nu)\to\frac{512}{15}\sqrt{\pi}$ as $\nu\to8\pi$.  
\end{enumerate}
\end{theorem}

In fact, the asymptotics in (2) of Theorem \ref{thm:canfreeenexp} are much more precise but we rewrite them in the present form to underline the following fact: in the limit when $\nu \to 8\pi$ the results above reproduce \eqref{eq:LHY} and \eqref{eq:3D_free_energy_SeiYin}. Indeed, in that limit we can replace $\hat{V}(0)$ by $8\pi a$ and the formulas in (1) and (2) of the theorem above reproduce the free energy expansion  \eqref{eq:3D_free_energy_SeiYin} while the expansion in (3) yields the Lee--Huang--Yang formula. 

Let us briefly comment on these results. As it is explained in, e.g., \cite[Appendix A]{LieSeiSolYng-05}, one can reproduce the Lee--Huang--Yang formula from Bogoliubov's original quadratic theory as long as one replaces $\hat{V}(0)$ by $8\pi a$. This substitution is motivated by the fact that $\hat{V}(0)$ is the first Born approximation to $8\pi a$. Then, in \cite{ErdSchYau-08}, it was noticed that including in the effective Hamiltonian some quartic terms that are obtained after the $c$-number substitution, allows one to reproduce the scattering length in the leading order term and get the right order of the second order term together with a constant that converges to the desired Lee--Huang--Yang constant as the interaction becomes weaker. In that sense the result in (3) of Theorem \ref{thm:canfreeenexp} does not come as a surprise. It is worth noticing, however, that since it provides also a lower bound, it shows that within the class of quasi-free states one cannot obtain the full Lee--Huang--Yang formula exactly (this is also the reason why the trial state used in \cite{YauYin-09} is more complicated). 

What is probably much more interesting about this functional is the result in Theorem \ref{thm:cancrittemp}. It shows that in the limit $\nu \to 8\pi$ the critical temperature of the interacting gas is given by 
\begin{equation}
\label{eq:Tc}
T_{\rm{c}}=T_{\rm{fc}}(1+1.49(\rho^{1/3}a)+o(\rho^{1/3}a)).
\end{equation}   
For general potentials, there has been a lot of debate about whether the linear dependence on $\rho^{1/3}a$ in \eqref{eq:Tc} is correct. Nonetheless, \eqref{eq:Tc} is still expected to hold true, at least up to the value of the constant 1.49 (we refer to \cite{SeiUel-09} for a comprehensive list of references regarding this issue; in fact, the upper bound on the critical temperature in \cite{SeiUel-09} is the only rigorous result on this issue). In a number of articles it was claimed that mean-field theories such as Bogoliubov theory cannot predict a change in the critical temperature. As Theorem \ref{thm:cancrittemp} shows, a variational, non-linear formulation of Bogoliubov theory indeed does show the predicted behaviour \eqref{eq:Tc} and with a constant 1.49 which is very close to 1.32 which is the constant obtained in Monte Carlo simulations for this problem (see, e.g., \cite{NhoLan-04}).

\subsection{Known results in the dilute limit in two dimensions.}
Historically, the two-dimensional Bose gas received attention much later than the three-dimensional system. This is probably because the non-interacting 2D Bose does not exhibit a phase transition at positive temperature. For systems with short range interactions, the celebrated Mermin--Wagner--Hohenberg theorem \cite{MerWag-66,Hohenberg-67} excludes continuous symmetry breaking and the presence of long-range order, therefore prohibiting BEC in its traditional sense. It has been realized only later, that 2D systems undergo the Kosterlitz-Thouless (KT) phase transition and quasi-long-range order can occur. In the context of Bose gases this is related to the concept of a quasi-condensate (\cite{Popov-83,MorCas-03}). In fact, it has been observed in experiments \cite{ClaRyuRamHelPhi-09} that a 2D Bose system undergoes two phase transitions: from the normal phase to a quasi-condensate without superfluidity and then the second one to a quasi-condensate with superfluidity, the latter being the KT phase transition. Since, as stated in Theorem \ref{thm:BECvsSF}, the Bogoliubov free energy functional does not differentiate between BEC and superfluidity, it sees only the KT phase transition. One should note here, that the experiment mentioned above dealt with a quasi 2D system. It is difficult to obtain truly homogeneous, translation invariant systems in the lab (see \cite{LopEigNavCleSmiHad-17} for recent progress in that direction). In \cite{NapReuSol-17} the critical temperature of a dilute, homogeneous 2D Bose gas has been determined. It turns out, in agreement with Schick's prediction \cite{Schick-71} that the relevant expansion parameter is, as already mentioned,
\bq \label{def:b}
b:=1/|\ln(\rho a^2)|\ll1
\eq
(here $a$ is the two dimensional scattering length, cf. Section \ref{sec:simplifiedfunctional}). We set
\bq \label{eq:apriori2Dtempbound}
\sqrt{T}a\leq C \rho^{1/2}a\ll 1
\eq
Due to dimensional consideration (notice $\hat{V}(0)$ is dimensionless) we assume
\bq\label{eq:Vexpansion}
\widehat{V}(p)=\widehat{V}(0)+Ca^2p^2+o(a^2p^2)
\eq
and the same for $\widehat{Vw}(p)$ (with, in general, other constants). With this  we have
\begin{theorem}
\label{thm:Crit2D}
Consider the canonical problem \eqref{def:canonicalminimization} in two dimensions. The critical temperature, defined by the properties $\rho_0= 0$ if $T>T_{\rm c}$, $\rho_0 > 0$ if $0\leq T<T_{\rm c}$, is
\begin{equation}
\label{critT}
T_{\rm{c}}=4\pi\rho\left(\frac{1}{\ln(\xi/4\pi b)}+o(1/\ln^2{b})\right),
\end{equation}
with $\xi\to 14.4$ as $\widehat{V}(0)\to 8\pi b$.
\end{theorem} 
This result improves upon several theoretical and numerical derivations in the literature (see, e.g., \cite{Popov-83,FisHoh-88,ProRueSvi-01}) as it provides a systematic way to compute the constant $\xi$. 
The proof of the statement above relies on a careful expansion of the free energy in the critical region.  The main steps of our approach will be presented in the next section.

\section{The simplified functional}\label{sec:simplifiedfunctional}
\subsection{Outline of proof.} Let us briefly recall the main idea behind the proof. In principle one could try to determine the minimum of the functional by analysing the associated Euler--Lagrange equations. However, when writing them out one will notice that the derivatives of the non-linear terms give rise to terms that are non-local, i.e. $\hat{V}\ast \gamma$ and $\hat{V}\ast \alpha$. This makes the Euler--Lagrange equations hard to analyze quantitatively. The idea is to show that, in the relevant region, the full functional can be effectively replaced by a simplified one that one can first solve explicitly in $\gamma$ and $\alpha$ and then minimize over $\rho_0$. To put it differently,
\begin{equation}\label{eq:approx-scheme}
\inf_{\substack{\text{$(\gamma,\alpha,\rho_0)$}\\\text{$\rho_0+\rho_\gamma=\rho$}}}
\mathcal{F}^{\rm can}\approx\inf_{0\leq\rho_0\leq\rho}\inf_{\substack{\text{$(\gamma,\alpha)$}\\\text{$\rho_\gamma=\rho-\rho_0$}}}\mathcal{F}^{\rm sim},
\end{equation}
where the infimum over $\gamma$ and $\alpha$ is calculated explicitly, then expanded properly in the dilute limit, and finally minimized in $\rho_0$.

The justification of this approximation will rely on several steps. First, we will replace the convolution term involving $\gamma$ with $\widehat{V}(0)\rho_\gamma^2$. We expect that the particles interact weakly in the dilute limit and it seems reasonable to assume that the system will behave like a free Bose gas to leading order. We therefore expect that the minimizing $\gamma$ is concentrated on a ball of radius $\sqrt{T}$ as in the non-interacting case which can be solved explicitly. On this scale $\widehat{V}(p)$ is approximately $\widehat{V}(0)$ justifying the replacement. 

Second, by introducing a trial function $\alpha_0$, we rewrite the convolution terms involving $\alpha$. This trial function will be expressed in terms of $\widehat{Vw}$, where $w$ is the solution to the scattering equation. Finally, we will also substitute $\widehat{V}$ by $\widehat{Vw}$ in the terms that are linear in $\gamma$ at the cost of a small error. These steps will rely on several a priori estimates.

We then minimize the simplified functional. We split the minimization in two steps: first one over $\gamma$ and $\alpha$ with the constraint that $\rho_0+\rho_\gamma=\rho$, followed by a minimization over $0\leq\rho_0\leq\rho$. The first step will lead to a useful class of minimizers $(\gamma^{\rho_0,\delta},\alpha^{\rho_0,\delta})$. The final minimization over $\rho_0$ will then follow. 

\subsection{Scattering equation in two dimensions.} As mentioned above, the simplified functional will involve a trial function $\alpha_0$ which will be related to the solution of the scattering equation or rather, to be more precise, its Fourier transform. Because of that we shall now review some properties of the two-dimensional scattering equation. It is given, in the sense of distributions, by
\begin{equation}\label{eq:scattering}
-\Delta w_0+\frac12 Vw_0=0,
\end{equation}
with the asymptotics $w_0(x) \approx \ln(|x|/a)$ when $|x|\gg R$ where $R$ is such that $\text{supp} V \subset B(0,R)$. Here $a$ is the scattering length given as  in \cite[Appendix C]{LieSeiSolYng-05}. Let us stress, that unlike in the three-dimensional case, in two dimensions the elegant characterization of the scattering length \eqref{def:scattlength3D} is not valid. In fact, in two dimensions 
we have 
\bq
\frac12 \int V w_0=\lim_{N\to \infty} \int_{B(0,N)} \Delta w_0=  \lim_{N\to \infty} \int_{|x|=N} \partial_r \ln(r/a)=2\pi. \label{eq:hatVphi0}
\eq
In the definition of $\alpha_0$ we will use $\hat{w}$ where $w=2bw_0$. In particular $\widehat{Vw}(0)=8\pi b$. Because of the logarithmic behaviour of $w$, the computation of its Fourier transform is more complicated than in three dimensions. In fact, we have
\begin{lemma}\label{lem:logFT}
Let $\epsilon=\frac{2}{a e^{\Gamma}}\exp(\frac{-1}{2b})$. Then the Fourier transform of $w$ is given by the distribution $\hat{w}$ of the form
$$\hat{w}=(2\pi)^2 \delta_0 -\hat{\varphi}$$
with
\begin{equation} \label{eq:hatwFT}
\hat{\varphi}(\phi)=\int_{|p|\leq \epsilon}\frac{\widehat{Vw}(p)\phi(p)-\widehat{Vw}(0)\phi(0)}{2p^2}dp+\int_{|p|>\epsilon}\frac{\widehat{Vw}(p)\phi(p)}{2p^2}dp.
\end{equation}
\end{lemma}
For the convenience of the reader we provide a proof of this fact in Appendix \ref{appendixA}. We stress that with the definition \eqref{def:b} we have that 
\begin{equation}
\epsilon = C\rho^{\frac12}  \label{eq:epsilonasympt}
\end{equation}
with $C=2/\exp(\Gamma)$.

\subsection{Derivation of the simplified functional.}  As mentioned before we expect the minimizing $\alpha$ to be related to the scattering solution. To this end we define
\begin{equation}
\label{alpha0}
\alpha_0:=(\rho_0+t_0)\widehat{w}-(2\pi)^2 \rho_0\delta_0,
\end{equation}
where $t_0\in [-\rho_0,0]$ is an additional parameter that will be tuned later on. With this definition we allow the scenario, that for small momenta $\alpha$ might be more complicated than the scattering solution. We approximate this part by a $\delta$-function and eventually optimize our approximation in $t_0$ (recall $\hat{w}$ has a $\delta_0$ in its definition). 

To approximate $\alpha$ by $\alpha_0$, we add and subtract terms to replace the convolution term with $\alpha$ by 
\begin{equation}
\int(\alpha-\alpha_0)(p)(\widehat{V}\ast(\alpha-\alpha_0))(p)dp,
\end{equation}
which we later show to be small for the minimizing $\alpha$. 
By doing this we have of course introduced terms involving $\widehat{V}\ast\alpha_0$, but
\begin{equation}\label{eq:Vconvalpha0}
(2\pi)^{-2}\widehat{V}*\alpha_0(p)=(\rho_0+t_0)\widehat{Vw}(p)-\rho_0\widehat{V}(p),
\end{equation}
so that no convolution terms remain in our functional. This has the added effect that $\widehat{V}$ gets replaced by $\widehat{Vw}$ in the term linear in $\alpha$, but this Fourier transform is well-defined and satisfies $\widehat{Vw}(0)=8\pi b$. To simplify the resulting functional, we make sure to obtain a similar replacement for the $\rho_0\int\widehat{V}\gamma$-term.

Motivated by these considerations we define
\begin{equation}
\label{errors}
\begin{aligned}
E_1&:=\frac12(2\pi)^{-4}\int(\alpha-\alpha_0)(p)(\widehat{V}\ast(\alpha-\alpha_0))(p)dp\\
E_2&:=\rho_0\left((2\pi)^{-2}\int\widehat{V}(p)\gamma(p)dp-\widehat{V}(0)\rho_\gamma\right)\\
E_3&:=-(\rho_0+t_0)\left((2\pi)^{-2}\int\widehat{Vw}(p)\gamma(p)dp-\widehat{Vw}(0)\rho_\gamma\right)\\
E_4&:=\frac12(2\pi)^{-4}\int\gamma(p)(\widehat{V}\ast\gamma)(p)dp-\frac12\widehat{V}(0)\rho_\gamma^2,
\end{aligned}
\end{equation}
and 
\begin{eqnarray}
\begin{aligned}
\label{eq:simplifiedfunctional}
  \mathcal{F}^{\rm{sim}}(\gamma,&\alpha,\rho_0)=(2\pi)^{-2}\left[\int p^2\gamma(p)dp+(\rho_0+t_0)\int\widehat{Vw}(p)(\gamma(p)+\alpha(p))dp\right]-TS(\gamma,\alpha)\\
&\quad+ 4\pi b(\rho_0+t_0)(3\rho_0-2\rho-t_0)+\widehat{V}(0)(\rho^2-\rho_0^2)\\  
&\quad+(2\pi)^{-2}(\rho_0+t_0)^2\int\frac{\widehat{Vw}(p)^2-\chi_{|p|\leq \epsilon}\widehat{Vw}(0)^2}{4p^2}dp.
\end{aligned}
\end{eqnarray}
This leads to the following 
\begin{lemma}\label{lem:FsimFcandiff}
With the definitions \eqref{errors} and \eqref{eq:simplifiedfunctional} we have
$$\mathcal{F}^{\rm{can}}(\gamma,\alpha,\rho_0)-\mathcal{F}^{\rm{sim}}(\gamma,\alpha,\rho_0)=E_1+E_2+E_3+E_4.$$
\end{lemma}
\begin{proof}
The proofs follows exactly the proof of \cite[Lemma 13]{NapReuSol-15b}. The main difference lies in the evaluation of the $\frac12(2\pi)^{-4}\int\alpha_0(p)(\widehat{V}*\alpha_0)(p)dp$ which, in particular, leads to the appearance of the last term in  \eqref{eq:simplifiedfunctional}. 
\end{proof}

The crucial property of the simplified functional is that, excluding the entropy term, it is linear in $\gamma$ and $\alpha$ and is much easier to analyze than the full functional.

\subsection{A priori estimates.}
We will now justify the approximation scheme \eqref{eq:approx-scheme}. To this end we shall derive several a priori estimates that will show that $\cF$ can indeed be approximated by $\cF^{\text{sim}}$ when one considers the minimization procedure. 

The main idea is the following: since we are considering a dilute gas, it is feasible to assume that the minimizer should to leading order behave like the one of the non-interacting problem
\begin{equation} \label{eq:nonintefunct}
\mathcal{F}_{0}(\gamma)=(2\pi)^{-2}\int p^2\gamma(p)dp-TS(\gamma,0).
\end{equation}
Its minimizer for given $\rho$ is
\begin{equation}
\label{sommin}
\gamma_{\mu(\rho)}(p)=\frac{1}{e^{(p^2-\mu(\rho))/T}-1},
\end{equation}
where $\mu(\rho)\leq0$ is such that $(2\pi)^{-2}\int\gamma_{\mu(\rho)}=\rho$. Since the integral diverges as $\mu(\rho)\to0$, this definition works for all $\rho\geq0$ (note the difference with respect to three spatial dimensions where a critical density occurs).

Let $\left(\gamma,\alpha,\rho_0=\rho-\rho_\gamma\right)$ be a minimizing triple for \eqref{def:canonicalfreeenergyfunctional} at a temperature $T$. Using the bound $\widehat{V}(p)\leq \widehat{V}(0)$ we find the following upper bound in terms of the free gas energy $\cF_0$
\begin{align}\label{eq:upperbound}
\cF^{\rm{can}}(\gamma,\alpha,\rho_0)\leq\ &\cF^{\rm{can}}(\gamma_{\mu(\rho)},0,0)\leq
\cF_0(\gamma_{\mu(\rho)})+\rho^2\widehat{V}(0).
\end{align}
We also have
\begin{equation}\label{eq:lowerbound}
\begin{aligned}
  \cF^{\rm{can}}(\gamma,\alpha,\rho_0)\geq\ &
  \cF_0(\gamma)+\frac12\widehat{V}(0)\rho^2+
  \rho_0(2\pi)^{-2}\int\widehat{V}(p)\gamma(p)dp \\&+
  \frac12(2\pi)^{-4}\iint\gamma(p)\widehat{V}(p-q)\gamma(q)dpdq
  -\frac12 \rho_0^2\widehat{V}(0),
\end{aligned}
\end{equation}
where we have first used that the entropy decreases if we replace $\alpha$
by 0 and then minimized over $\alpha$, finding the minimizer
$\alpha=-(2\pi)^{2}\rho_0\delta_0$. We conclude  
\begin{equation}\label{eq:1stenergy}
\cF_0(\gamma_{\mu(\rho_\gamma)})\leq \cF_0(\gamma)\leq \cF_0(\gamma_{\mu(\rho)})+\rho^2\widehat{V}(0).
\end{equation}
We will use this to give an estimate on the integral of $\gamma$ in a
region $|p|>\xi$, where $\xi$ is to be chosen below.  We shall use the following result whose proof is identical as in the three dimensional case (\cite[Lemma 18]{NapReuSol-15b}).

\begin{lemma}[A priori kinetic energy bound]\label{lm:apriorikinetic}
If for some $Y>0$ the function $\gamma$ satisfies
$\cF_0(\gamma)\leq \cF_0(\gamma_{\mu(\rho_\gamma)})+Y$, then for all $\xi$ with $\xi^2>8T$ we have 
$$
\frac12  (2\pi)^{-2}\int_{|p|>\xi} p^2\gamma(p)dp \leq Y+CT^{2}e^{-\xi^2/4T}.
$$
\end{lemma}
 
Using $\cF_0(\gamma_{\mu(\rho_\gamma)})\geq \cF_0(\gamma_{\mu(\rho)})$ in \eqref{eq:1stenergy}, we can use this lemma with $Y=\rho^2\widehat{V}(0)$ to conclude that 
\bq
\begin{aligned}
\iint_{|p-q|>2\xi} \gamma(p)\widehat{V}(p-q)\gamma(q)dpdq
\leq\ & C \widehat{V}(0)\rho\int_{|p|>\xi}\gamma(p)dp \\
\leq \ &C \widehat{V}(0)\rho(\rho^2\widehat{V}(0)+T^{2}e^{-\xi^2/4T})\xi^{-2}.\label{eq:gvg1}
\end{aligned}
\eq
We choose $\xi=a^{-1}(\rho^{1/2} a)^{1/2}$. Then, using \eqref{eq:apriori2Dtempbound},  $\xi^2/T\geq C^{-1}(\rho^{1/2} a)^{-1/2}\gg1$ and we find 
\begin{equation}
\label{lem:derivsimplified}
\iint_{|p-q|>2\xi}
\gamma(p)\widehat{V}(p-q)\gamma(q)dpdq\leq
C\rho^3\widehat{V}(0)^2\xi^{-2}\leq C\rho^2 (\rho^{1/2} a).
\end{equation}
Of course, the same bound holds if $\widehat{V}(p-q)$ is replaced by $\widehat{V}(0)$.
On the other hand we also have
\begin{equation}
\label{estm2}
\iint_{|p-q|<2\xi} \gamma(p)|\widehat{V}(p-q)-\widehat{V}(0)|\gamma(q)dpdq
\leq 
C\xi^2 \rho^2 a^2 \leq C\rho^2 (\rho^{1/2} a).
\end{equation}
where we used \eqref{eq:Vexpansion}. For the same choice of $\xi$:
\begin{equation}
\begin{aligned}
\left|\int\gamma(p)\widehat{V}(p)dp \right.&\left. 
-\widehat{V}(0)\int\gamma(p)dp\right| \leq \left|\left(\int\limits_{|p|\leq \xi}+\int\limits_{|p|> \xi}\right)\gamma(p)\left(\widehat{V}(p)
- \widehat{V}(0)\right)dp\right| \\
 &\leq C\xi^2 a^2  \int_{|p|\leq \xi}\gamma(p)dp+C\widehat{V}(0)\xi^{-2}\int_{|p|> \xi}p^2 \gamma(p)dp\\
&\leq C\rho a^2 \xi^2+C \xi^{-2}(\rho^2 +T^{2}e^{-\xi^2/4T})\leq C\rho (\rho^{1/2} a). \label{estimongammV}
\end{aligned}  
\end{equation}
The same bounds hold for $\widehat{Vw}$  which, because $V$ is compactly supported and smooth, is well defined and smooth. We have thus shown the following result.

\begin{proposition}[A priori estimates]
\label{lem:aprioriest}
Any minimizing triple $(\gamma,\alpha,\rho_0)$ with density $\rho=\rho_\gamma+\rho_0$ and temperature $T$
satisfying $T<D \rho$ obeys the estimates
$$
\left|(2\pi)^{-4}\iint\gamma(p)\widehat{V}(p-q)\gamma(q)dpdq
- \widehat{V}(0)\rho_\gamma^2\right|\leq C\rho^2 (\rho^{1/2} a),
$$
$$
\left|(2\pi)^{-2}\int\gamma(p)\widehat{V}(p)dp
- \widehat{V}(0)\rho_\gamma\right|\leq C\rho (\rho^{1/2} a),
$$
where the constant $C$ depends on $D$ and the potential $V$. This second inequality holds also with $\widehat{V}$ replaced by $\widehat{Vw}$.
\end{proposition}

\subsection{Minimization of the simplified functional.} 
Since $E_1$ is non-negative, it follows from Lemma \ref{lem:FsimFcandiff} that for any triple $(\gamma,\alpha,\rho_0)$ we have
$$\mathcal{F}^{\rm{sim}}(\gamma,\alpha,\rho_0)+E_2+E_3+E_4\leq\mathcal{F}^{\rm{can}}(\gamma,\alpha,\rho_0)=\mathcal{F}^{\rm{sim}}(\gamma,\alpha,\rho_0)+E_1+E_2+E_3+E_4.$$
The first inequality together with the a priori estimates in Lemma \ref{lem:aprioriest} implies that any potential minimizer satisfies 
\begin{equation}\label{eq:apriorilowerboundfull}
\begin{aligned}
\cF^{\rm can}(\gamma,\alpha,\rho_0)&\geq\cF^{\rm sim}(\gamma,\alpha,\rho_0)+(E_2+E_3+E_4)(\gamma,\alpha,\rho_0)\\
&\geq\cF^{\rm sim}(\gamma^{\rho_0,\delta},\alpha^{\rho_0,\delta},\rho_0)-O(\rho^2(\rho^{1/2} a)).\end{aligned}
\end{equation}
where $(\gamma^{\rho_0,\delta},\alpha^{\rho_0,\delta},\rho_0)$ is the minimizing triple for the simplified functional. Obviously, for a minimizing triple of the canonical functional $(\gamma,\alpha,\rho_0)$ we also have 
\begin{equation} \label{eq:upperbounderrorterms}
\mathcal{F}^{\rm{can}}(\gamma,\alpha,\rho_0)\leq \mathcal{F}^{\rm{sim}}(\gamma^{\rho_0,\delta},\alpha^{\rho_0,\delta},\rho_0)+(E_1+E_2+E_3+E_4)(\gamma^{\rho_0,\delta},\alpha^{\rho_0,\delta},\rho_0).
\end{equation}
Because of that, our goal will now be to minimize the simplified functional. We note that the minimization problem can be rewritten as 
\bq
\begin{aligned}
&\inf_{(\gamma,\alpha,\rho_0),\ \rho_\gamma+\rho_0=\rho}\mathcal{F}^{\rm sim}(\gamma,\alpha,\rho_0)
=\inf_{0\leq\rho_0\leq\rho}\Big[\inf_{(\gamma,\alpha),\ \rho_\gamma=\rho-\rho_0}\mathcal{F}^{\rm s}(\gamma,\alpha,\rho_0) +\widehat{V}(0)(\rho^2-\rho_0^2) \\  &\qquad  +4\pi b(\rho_0+t_0)(3\rho_0-2\rho-t_0) \Big],
\end{aligned} \label{eq:simplifiedfinal}
\eq
with
\begin{equation*}
\begin{aligned}
\mathcal{F}^{\rm{s}}(\gamma,\alpha,\rho_0)&=(2\pi)^{-2}\left[\int p^2\gamma(p)dp+(\rho_0+t_0)\int\widehat{Vw}(p)(\gamma(p)+\alpha(p))dp\right]-TS(\gamma,\alpha)\\
&+(2\pi)^{-2}(\rho_0+t_0)^2\int\frac{\widehat{Vw}(p)^2-\chi_{|p|\leq \epsilon}\widehat{Vw}(0)^2}{4p^2}dp.
\end{aligned}
\end{equation*}
The advantage of considering $\cF^{\rm s}$ is that it can be solved explicitly. To this end let us define
\begin{eqnarray*}
G(p)&=&T^{-1}\sqrt{(p^2+\delta+(\rho_0+t_0)
  \widehat{Vw}(p))^2-((\rho_0+t_0) \widehat{Vw}(p))^2}\\&=&
T^{-1}\sqrt{(p^2+\delta)^2+
  2(p^2+\delta)(\rho_0+t_0)\widehat{Vw}(p)}.
\end{eqnarray*}
We have the following lemma whose proof is exactly the same as in the three dimensional case.
\begin{lemma} 
\label{prop:simplfunctsol}
Let $\delta\geq0$, $\rho_0\geq0$ and $-\rho_0\leq t_0 \leq 0$. The minimizer of 
\[
\inf_{(\gamma,\alpha)}\left[\mathcal{F}^{\rm s}(\gamma,\alpha,\rho_0)+\delta\int\gamma\right]
\]
is given by
\bq
\begin{aligned}
\gamma^{\rho_0,\delta}&=\frac{\beta}{TG}(p^2+\delta+(\rho_0+t_0) \widehat{Vw}(p))-\frac12 \\
\alpha^{\rho_0,\delta}&=-\frac{\beta}{TG}(\rho_0+t_0)\widehat{Vw}(p),
\end{aligned} \nn
\eq
with $G$ as above and 
$$\beta(p)=(e^{G(p)}-1)^{-1}+\frac12.$$
The minimum is
\[
\begin{aligned}
\cF^{\rm s}&(\gamma^{\rho_0,\delta},\alpha^{\rho_0,\delta},\rho_0)+\delta\int\gamma^{\rho_0,\delta} =\quad (2\pi)^{-2} T\int\ln(1-e^{-G(p)})dp \\
&+(2\pi)^{-2} \frac12\int\Big[\sqrt{(p^2+\delta)^2+2(p^2+\delta)(\rho_0+t_0) \widehat{Vw}(p)}-(p^2+\delta+(\rho_0+t_0) \widehat{Vw}(p))\\
&+(2\pi)^{-2}(\rho_0+t_0)^2\int\frac{\widehat{Vw}(p)^2-\chi_{|p|\leq \epsilon}\widehat{Vw}(0)^2}{4p^2}dp\Big]dp.
\end{aligned}
\]
\end{lemma}

\section{Ground state energy expansion - proof of Theorem \ref{thm:mainresult}} \label{sec:2DfreeenergylowT}
In this section we will prove the main result of this paper. We divide the proof into two parts. First, we will derive the leading order estimates from the analysis of $\cF^{\rm{sim}}$. In the second part, we will show that the error terms in \eqref{eq:upperbounderrorterms} are indeed of lower order.
\subsection{Analysis of $\cF^{\rm{sim}}$.}  We will analyze the simplified functional in the ground state, which implies in particular $T=0$ (to be more precise, after computing the relevant quantities we take the limit $T\to 0$). As in the three dimensional analysis, we set $t_0=0$. Under these conditions, the expression in \eqref{eq:simplifiedfinal} can be rewritten  in the following way
\bq
\begin{aligned}
&\inf_{(\gamma,\alpha,\rho_0),\ \rho_\gamma+\rho_0=\rho}\mathcal{F}^{\rm sim}(\gamma,\alpha,\rho_0)
=\inf_{0\leq\rho_0\leq\rho}\Big[ 4\pi b\rho^2\left(3\frac{\rho_0^2}{\rho^2}-2\frac{\rho_0}{\rho}\right)+\widehat{V}(0)\rho^2\left(1-\frac{\rho_0^2}{\rho^2}\right) \\  &\qquad\qquad+\inf_{(\gamma,\alpha),\ \rho_\gamma=\rho-\rho_0}\mathcal{F}^{\rm s}(\gamma,\alpha,\rho_0)\Big].
\end{aligned} \label{eq:simplifiedT=0t0}
\eq
Using  Lemma \ref{prop:simplfunctsol} at $T=0$ we have 
\begin{equation*}\label{eq:cFrmstoanalyze}
\begin{aligned} 
&\inf_{(\gamma,\alpha),\ \rho_\gamma=\rho-\rho_0}\mathcal{F}^{\rm s}(\gamma,\alpha,\rho_0)= -\delta\rho_{\gamma^{\rho_0,\delta}}+(2\pi)^{-2}\rho_0^2\int\frac{\widehat{Vw}(p)^2-\chi_{|p|\leq \epsilon}\widehat{Vw}(0)^2}{4p^2}dp \\
&+ (2\pi)^{-2}\frac{1}{2}\int\Big[\sqrt{(p^2+\delta)^2+2(p^2+\delta)\rho_0 \widehat{Vw}(p)} -(p^2+\delta+\rho_0 \widehat{Vw}(p))\Big]dp .
\end{aligned}
\end{equation*}
We set $\delta=d\rho_0 b$. After a change of variables ($p\mapsto p\sqrt{\rho_0 b}$) we get
\begin{equation}
\begin{aligned} \label{eq:FsT=0}
&\inf_{(\gamma,\alpha),\ \rho_\gamma=\rho-\rho_0}\mathcal{F}^{\rm s}(\gamma,\alpha,\rho_0)= -db\rho_0 \rho_{\gamma^{\rho_0,\delta}}+(2\pi)^{-2}\rho_0^2\int\frac{\widehat{Vw}(\sqrt{\rho_0 b}p)^2-\chi_{|p|\leq \frac{\epsilon}{\sqrt{\rho_0 b}}}\widehat{Vw}(0)^2}{4p^2}dp \\
&+\frac{(\rho_0 b)^2}{2(2\pi)^2}\int\Big[\sqrt{(p^2+d)^2+2(p^2+d)8\pi\frac{\widehat{Vw}(\sqrt{\rho_0 b}p)}{8\pi b}} -(p^2+d+ 8\pi\frac{\widehat{Vw}(\sqrt{\rho_0 b}p)}{8\pi b})\Big]dp .
\end{aligned}
\end{equation}
Similarly,
\begin{equation}
\begin{aligned}
\label{eq:rhogammasimplified}
\rho_{\gamma^{\rho_0,\delta}}&=(2\pi)^{-2}\frac{\rho_0 b}{2}\int\left(\frac{p^2+d+8\pi\frac{\widehat{Vw}(\sqrt{\rho_0 b}p)}{8\pi b}}{\sqrt{(p^2+d)^2+2(p^2+d)8\pi \frac{\widehat{Vw}(\sqrt{\rho_0 b}p)}{8\pi b}}}-1\right)dp.
\end{aligned}
\end{equation}
\begin{lemma}
Let $\rho_{\gamma^{\rho_0,\delta}}$ be given as in \eqref{eq:rhogammasimplified}. Then
\begin{equation} \label{eq:rhogammaasymp}
\rho_{\gamma^{\rho_0,\delta}} =\frac{\rho_0 b}{4\pi}\left[4\pi-\frac12\left(\sqrt{d(d+16\pi)}-d\right)\right]+ o(\rho_0 b)=: C(d) \rho_0 b + o(\rho_0 b)
\end{equation}
where $C(d)\in [0,1]$ for $d\geq 0$ and  $C(d)= O(\frac1d)$ as $d\to\infty$.
\end{lemma}
\begin{proof}
Since $\frac{|\widehat{Vw}(\sqrt{\rho_0 b}p)|}{ 8\pi b}\leq 1$ and also converges to $1$ as $b\to 0$, we can use a dominated convergence argument as in  \cite[Lemma 26]{NapReuSol-15b}) to obtain
\begin{equation*} 
\begin{aligned}
\rho_{\gamma^{\rho_0,\delta}} &=\frac{\rho_0 b}{4\pi}\int_0^\infty\left(\frac{p^2+d+8\pi}{\sqrt{(p^2+d)^2+2(p^2+d)8\pi}}-1\right)pdp+ o(\rho_0 b)\\
&=\frac{\rho_0 b}{4\pi}\left[4\pi-\frac12\left(\sqrt{d(d+16\pi)}-d\right)\right]+ o(\rho_0 b)= C(d)\rho_0 b + o(\rho_0 b).
\end{aligned}
\end{equation*}
The result follows from a straightforward analysis of $C(d)$.
\end{proof}
Note that \eqref{eq:rhogammaasymp} implies $\delta\rho_{\gamma^{\rho_0,\delta}}=dC(d) \rho^2 b^2 +o(\rho^2 b^2).$ for any $d$. We also have the following 
\begin{corollary}\label{cor:leadingorder}
Let $\hat{V}(0)=\nu b$ and let $d$ be fixed. Then  
\begin{equation}
4\pi b\rho^2\left(3\frac{\rho_0^2}{\rho^2}-2\frac{\rho_0}{\rho}\right)+\widehat{V}(0)\rho^2\left(1-\frac{\rho_0^2}{\rho^2}\right)=4\pi b\rho^2 +\rho^2 b^2\left(  2\nu C(d)-16\pi C(d)   \right)+  o(\rho^2 b^2).  
\end{equation}
\end{corollary}
\begin{proof}
Since $\rho=\rho_0+\rho_{\gamma^{\rho_0,\delta}}$, using \eqref{eq:rhogammaasymp} we have $\rho=\rho_0+C(d)\rho_0 b+o(b)=\rho_0(1+C(d)b)+o(b)$. Thus 
\begin{equation}
\frac{\rho_0}{\rho}=\frac{1}{1+C(d)b}=1-C(d)b +o(b). \label{eq:rho0rhorelation}
\end{equation}
The statement follows from a direct computation and the assumption that $\hat{V}(0)=\nu b$.
\end{proof}
Notice that Corollary \ref{cor:leadingorder} gives the leading order term in the statement of Theorem \ref{thm:mainresult}. To determine the second order term we will now analyze $\cF^{\rm{s}}$. Recall \eqref{eq:FsT=0}. Since
\begin{equation} \label{eq:nonintegrable}
\begin{aligned}
\frac{(\rho_0 b)^2}{2(2\pi)^2}\int\Big[\sqrt{(p^2+d)^2+2(p^2+d)8\pi\frac{\widehat{Vw}(\sqrt{\rho_0 b}p)}{8\pi b}} -(p^2+d+ 8\pi\frac{\widehat{Vw}(\sqrt{\rho_0 b}p)}{8\pi b})\Big]dp
\end{aligned}
\end{equation}
is non-integrable (at infinity), we make the integral convergent by adding to it the large $p$ term
from the last term in the first line of \eqref{eq:FsT=0}. Recall that the cut-off $\epsilon$ given in Lemma \ref{lem:logFT} satisfies $\epsilon=O(\sqrt{\rho})$. We have
\begin{equation}
\begin{aligned}
& (2\pi)^{-2}\rho_0^2\int\frac{\widehat{Vw}(\sqrt{\rho_0 b}p)^2-\chi_{|p|\leq \frac{\epsilon}{\sqrt{\rho_0 b}}}\widehat{Vw}(0)^2}{4p^2}dp= \\
& \frac{\rho_0^2}{(2\pi)^2}\int_{|p|\leq \frac{\epsilon}{\sqrt{\rho_0 b}}}\frac{\widehat{Vw}(\sqrt{\rho_0 b}p)^2-\widehat{Vw}(0)^2}{4p^2}dp+\frac{\rho_0^2}{(2\pi)^2}\int_{|p|> \frac{\epsilon}{\sqrt{\rho_0 b}}}\frac{\widehat{Vw}(\sqrt{\rho_0 b}p)^2}{4p^2}dp=:I_1+I_2.
\end{aligned}
\end{equation}
Using \eqref{eq:Vexpansion} (for $\widehat{Vw}$), the first term can be bounded by
\begin{equation}\label{eq:I1bound}
|I_1|\leq \frac{\rho_0^2}{(2\pi)^2}\int_{|p|\leq \frac{\epsilon}{\sqrt{\rho_0 b}}}\frac{|\widehat{Vw}(\sqrt{\rho_0 b}p)^2-\widehat{Vw}(0)^2|}{4p^2}dp\leq Cb \rho_0^2 (\epsilon^2 a^2)
\end{equation}
which using \eqref{eq:epsilonasympt} and the definition of $b$ implies that the contribution of $I_1$ is negligible. Let us now analyze \eqref{eq:nonintegrable} together with $I_2$. In particular 
\begin{equation} \label{eq:I<I>}
\begin{aligned}
& \eqref{eq:nonintegrable}+I_2  \\
&= \frac{(\rho_0 b)^2}{2(2\pi)^2}\int_{|p|\leq \frac{\epsilon}{\sqrt{\rho_0 b}}}\Big[\sqrt{(p^2+d)^2+2(p^2+d)8\pi\frac{\widehat{Vw}(\sqrt{\rho_0 b}p)}{8\pi b}} -(p^2+d+ 8\pi\frac{\widehat{Vw}(\sqrt{\rho_0 b}p)}{8\pi b})\Big]dp \\
&+\frac{(\rho_0 b)^2}{2(2\pi)^2}\int_{|p|> \frac{\epsilon}{\sqrt{\rho_0 b}}}\Big[\sqrt{(p^2+d)^2+2(p^2+d)8\pi\frac{\widehat{Vw}(\sqrt{\rho_0 b}p)}{8\pi b}} -(p^2+d+ 8\pi\frac{\widehat{Vw}(\sqrt{\rho_0 b}p)}{8\pi b})\\
 &  \qquad\qquad\qquad\qquad\qquad\qquad+64\pi^2\frac{\widehat{Vw}^2(\sqrt{\rho_0 b}p)}{128\pi^2 b^2 p^2 }\Big]dp =: I_< (d)+I_>(d). 
\end{aligned}
\end{equation}
The dominant contribution comes from $I_<$ which we will analyze first.

\begin{lemma} \label{lem:I<}
Let $I_< (d)$ be defined as in \eqref{eq:I<I>}. Then
\begin{equation} 
I_< (d) =\frac{(\rho_0 b)^2}{2(2\pi)^2}\left(\int_{|p|\leq \frac{\epsilon}{\sqrt{\rho_0 b}}}\Big[\sqrt{(p^2+d)^2+2(p^2+d)8\pi} -(p^2+d+ 8\pi)\Big]dp + o(b)\right). \label{eq:I_<dintegral}
\end{equation} 
In particular, 
\begin{equation}\label{eq:exactI<}
\begin{aligned}
I_< (d) =&\frac{(\rho_0 b)^2}{4\pi}\Big[\frac{d^2}{4}+8\pi^2+4 \pi d-16 \pi^2 \ln \left(\frac{2\epsilon^2}{\rho_0 b} \right)- \frac{1}{4} d \sqrt{d (d+16 \pi )}-2 \pi  \sqrt{d (d+16 \pi )} \\
& +16 \pi ^2 \ln \left(d+\sqrt{d (d+16 \pi )}+8 \pi \right) + O(b)\Big].
\end{aligned}
\end{equation}
\end{lemma}
\begin{proof}
The proof is similar (and simpler) to the proof of  \cite[Lemma 27]{NapReuSol-15b}). Let us briefly sketch it. We define 
$$f(p,t):=\sqrt{(p^2+d)^2+2(p^2+d)8\pi t} -(p^2+d+ 8\pi t).$$
We estimate 
$$|f\left(p,\frac{\widehat{Vw}(\sqrt{\rho_0 b}p)}{8\pi b}\right)-f(p,1)|\leq \sup_{t \in [\frac{\widehat{Vw}(\sqrt{\rho_0 b}p)}{8\pi b}, 1]} |\partial_t f(p,t)| \Big|\frac{\widehat{Vw}(\sqrt{\rho_0 b}p)}{8\pi b}-1\Big|$$
For $|p|\leq \epsilon/\sqrt{\rho_0 b}$ we have
\begin{equation}\label{eq:estimatet-1}
|\frac{\widehat{Vw}(\sqrt{\rho_0 b}p)}{8\pi b}-1\Big|=\Big|\frac{\widehat{Vw}(\sqrt{\rho_0 b}p)-\widehat{Vw}(0)}{8\pi b}\Big| \leq C\frac{\epsilon^2 a^2}{b}\leq C b^m
\end{equation}
for any $m\in \mathbb{N}.$ The last estimate allows, in particular, for a very crude bound on $|\partial_t f(p,t)|$ which can be easily shown to be $|\partial_t f(p,t)|\leq 8\pi$. The final result follows from the fact that $\epsilon/\sqrt{\rho_0 b}=O(b^{-1/2})$. Equation \eqref{eq:exactI<} follows from a direct computation.
\end{proof}
It follows from Lemma \ref{lem:I<} that for small $b$ the integral $I_<(d)=O(\rho^2 b^2 \ln b)$. To determine $d$ we need to collect all terms of order $\rho^2 b^2$. Before that, let us show that there are no more terms of order $\rho^2 b^2$ left. 
\begin{lemma} \label{lem:I>}
Let $I_> (d)$ be defined as in \eqref{eq:I<I>}. Then
\begin{equation*} 
I_> (d) =\frac{(\rho_0 b)^2}{2(2\pi)^2}\left(\int_{|p|\geq \frac{\epsilon}{\sqrt{\rho_0 b}}}\Big[\sqrt{((p^2+d)^2+16\pi (p^2+d)} -(p^2+d+ 8\pi)+\frac{32\pi^2}{p^2}\Big]dp + O(b)\right). \label{eq:I_>dintegral}
\end{equation*} 
In particular, 
$$I_> (d) =o\big(\rho^2 b^2 \big).$$
\end{lemma}
The proof of this statement is similar to the proof of Lemma \ref{lem:I<}. In this case, however, the term \eqref{eq:estimatet-1} can only be bounded by a constant. On the other hand, since $1/p^2=O(b)\ll 1$, one can estimate $\sup |\partial_t f(p,t)|$ by an integrable function. We omit the details. 

We arrive at the following result which yields the desired expansion in Theorem \ref{thm:mainresult}.
\begin{corollary}\label{cor:final}
Let $\mathcal{F}^{\rm sim}$ be defined as in \eqref{eq:simplifiedT=0t0} and let $\hat{V}(0)= \nu b$. Then 
$$\inf_{(\gamma,\alpha,\rho_0),\ \rho_\gamma+\rho_0=\rho}\mathcal{F}^{\rm sim}(\gamma,\alpha,\rho_0)=4\pi \rho^2 b + 4\pi \rho^2 b^2 \ln b + \left(\inf_{d\geq 0} C_\nu (d)\right)\rho^2 b^2 + o \big(\rho^2 b^2 \big)$$
where $C_\nu (d)$ is given by \eqref{eq:Cvd}.
\end{corollary}
\begin{proof}
The proof follows from a straightforward calculation using Corollary \ref{cor:leadingorder}, Lemma \ref{lem:I<} and Lemma \ref{lem:I>}. Furthermore, using \eqref{eq:epsilonasympt} and \eqref{eq:rho0rhorelation}, we notice that
$$\ln\left(\frac{2\epsilon^2}{\rho_0} \right)=  \ln\left(\frac{2\epsilon^2}{\rho} \right)+\ln\left(\frac{\rho}{\rho_0} \right)= \ln\left(\frac{8}{e^{2\Gamma}} \right)+O(b).   $$
\end{proof}

\subsection{Estimates on remaining error terms.} It follows from \eqref{eq:apriorilowerboundfull} that 
$$\cF^{\rm can}(\gamma,\alpha,\rho_0)\geq\cF^{\rm sim}(\gamma^{\rho_0,\delta},\alpha^{\rho_0,\delta},\rho_0)-o(\rho^2 b^2).$$
Thus we only need to analyze the terms in \eqref{eq:upperbounderrorterms}. Recalling the definitions of \eqref{errors} and using $\hat{V}(0)\leq Cb$ as well as \eqref{eq:rhogammaasymp}, we immediately see that 
$$\big|E_i (\gamma^{\rho_0,\delta},\alpha^{\rho_0,\delta},\rho_0))\big| \leq C\rho^2 b^2, \qquad i=2,3,4.$$
which shows that these terms are indeed of lower order with respect to the expansion in Corollary \ref{cor:final}. It remains to analyze $E_1 (\gamma^{\rho_0,\delta},\alpha^{\rho_0,\delta},\rho_0)$. To this we notice that we notice that, at $T=0$ and for $d=0$, 
$$\alpha^{\rho_0,\delta}-\alpha_0=-\frac12 \frac{\rho_0 \widehat{Vw(p)}}{\sqrt{(p^2+\delta)^2+2(p^2+\delta) \rho_0 \widehat{Vw(p)}}}+\rho_0\hat{\varphi}$$
where the distribution $\hat{\varphi}$ is given in Lemma \ref{lem:logFT}. The action of $\hat{\varphi}$ can be split into two terms, the singular given by the first term in \eqref{eq:hatwFT} and the regular one given by the second term in  \eqref{eq:hatwFT}. Accordingly, we write $\hat{\varphi}=:\hat{\varphi}_1+\hat{\varphi}_2$. We rewrite
$$\alpha^{\rho_0,\delta}-\alpha_0=-\frac12 \frac{\rho_0 \widehat{Vw(p)}}{\sqrt{(p^2+\delta)^2+2(p^2+\delta) \rho_0 \widehat{Vw(p)}}}+\rho_0\hat{\varphi}_2+\rho_0\hat{\varphi}_1=: \rho_0\tilde{\alpha}+\rho_0\hat{\varphi}_1.$$
Then $E_1 (\gamma^{\rho_0,\delta},\alpha^{\rho_0,\delta},\rho_0))$
involves three terms:
\begin{equation*}
A_1=\rho_0^2\int\tilde{\alpha}(p)(\widehat{V}\ast\tilde{\alpha})(p)dp, 
A_2=\rho_0^2 \int \tilde{\alpha}(p)(\widehat{V}\ast\hat{\varphi}_1)(p)dp, 
A_3=\rho_0^2 \int \hat{\varphi}_1(p) (\widehat{V}\ast\hat{\varphi}_1)(p)dp.
\end{equation*}
The terms $A_2$ and $A_3$ are of lower order. This can seen from the fact that the action of $\hat{\varphi}_1$ restricts to $|p|\leq \epsilon$. Consequently, similarly to the bound in \eqref{eq:I1bound}, both terms can be bounded by $\rho^2 (\epsilon a)$ which is $o(\rho^2 b^2).$

It remains to analyze $A_1$. Note, that since $\hat{V}(0)= \nu b$, we have that 
\begin{equation} \label{eq:A_1generalbound}
A_1\leq \rho^2 b (\int |\tilde{\alpha}|)^2.
\end{equation}
We will now estimate $\int |\tilde{\alpha}|$. We have
\begin{equation*}
\begin{aligned}
\int |\tilde{\alpha}| (p) dp &=  \int_{|p|>\epsilon}\Big|\Big[\frac{-\widehat{Vw}(p)}{2\sqrt{(p^2+\delta)^2+2(p^2+\delta) \rho_0 \widehat{Vw}(p)}}+\frac{\widehat{Vw}(p)}{2p^2}\Big]\Big|dp \\
& + \int_{|p|\leq \epsilon} \Big|\frac{-\widehat{Vw}(p)}{2\sqrt{(p^2+\delta)^2+2(p^2+\delta) \rho_0 \widehat{Vw}(p)}}\Big|dp =:A_1^>+A_1^<.
\end{aligned}
\end{equation*}
As in the analysis before, we rescale $\delta$ and the $p$ variable and get 
\begin{equation*}
\begin{aligned}
A_1^< &= b\int_{|p|\leq \frac{\epsilon}{\sqrt{\rho_0 b}}} \Big|\frac{-8\pi \frac{\widehat{Vw}(\sqrt{\rho_0 b}p)}{8\pi b}}{2\sqrt{(p^2+d)^2+16\pi (p^2+d) \frac{\widehat{Vw}(\sqrt{\rho_0 b}p)}{8\pi b}}}\Big|dp,\\
A_1^> &= b\int_{|p|>\frac{\epsilon}{\sqrt{\rho_0 b}}}\Big|\Big[\frac{-8\pi\frac{\widehat{Vw}(\sqrt{\rho_0 b}p)}{8\pi b}}{2\sqrt{(p^2+d)^2+16\pi (p^2+d)  \frac{\widehat{Vw}(\sqrt{\rho_0 b}p)}{8\pi b}}}+\frac{8\pi\frac{\widehat{Vw}(\sqrt{\rho_0 b}p)}{8\pi b}}{2p^2}\Big]\Big|dp.
\end{aligned}
\end{equation*}
Notice, that the term coming from $\hat{\varphi}_2$ makes the outer integral converge. In particular, like in Lemmas \ref{lem:I<} and \ref{lem:I>}, one can show that to leading order in $b$ these integrals can be replaced by those with $\widehat{Vw}(\sqrt{\rho_0 b}p)/(8\pi b)=1$. In particular, a direct computation shows that
\begin{equation*}
 b\int_{|p|\leq \frac{\epsilon}{\sqrt{\rho_0 b}}} \frac{8\pi }{2\sqrt{(p^2+d)^2+16\pi (p^2+d) }}dp=O(b\ln b)
\end{equation*}
which implies 
$$A_1^<=O(b\ln b)  .$$
Similarly, one can show that $A_1^>=o(b\ln b)$. Combining these estimates with \eqref{eq:A_1generalbound}, we conclude that 
$$A_1\leq C\rho^2 b \Big[(b\ln b)^2 +o((b\ln b)^2)\Big]=C \rho^2 b^2 \ln b \Big[(b\ln b)+o(b\ln b)\Big] =o(\rho^2 b^2)$$ 
and thus is subleading with respect to the terms in the statement of Theorem \ref{thm:mainresult}. This ends the proof.

\appendix
\section{Proof of Lemma \ref{lem:logFT}} \label{appendixA}
\begin{proof}
Let $\varphi=1-w$. From the scattering equation  it follows that
\bq \label{eq:hatwawayfrom 0}
\hat{\varphi}=\frac12 \frac{\widehat{Vw}(p)}{p^2} \,\,\, \text{on} \,\,\, \R^2\setminus\{0\}.
\eq 
Since we are considering distributions, we have to determine the singular term at the origin. To this end we notice that
$$\varphi=1-2b\ln(r/a)+\tilde{w}$$
with $\tilde{w}\in L^1 (\R^2)\cap L^2 (\R^2)$ (this follows from the fact that $\tilde{w}$ describes the behaviour of the $L^2$- function $w$ for $|x|<R$). Thus
\bq \label{eq:scatt_fourier_ln}
\hat{\varphi}=(2\pi)^2(1+ 2b\ln(a))\delta_0-2b\widehat{\ln(r)}+\widehat{\tilde{w}}
\eq
where $\delta_0$ is the Dirac-delta distribution at the origin. Clearly $\widehat{\tilde{w}}\in L^\infty (\R^2)\cap L^2 (\R^2)$. We now need to compute the Fourier transform of the tempered distribution given by $\ln|x|$ on $\R^2$. We will denote it by $\mathcal{L}$. Introducing the distribution $\mathcal{P}$ given by
$$\mathcal{P}(\phi)=-2\pi\left(\int_{|p|\leq 1}\frac{\phi(p)-\phi(0)}{p^2}dp+\int_{|p|>1}\frac{\phi(p)}{p^2}dp\right)$$
we easily compute that 
$$(-\Delta \mathcal{P})(\check{\phi})=-2\pi \phi(0).$$
Here $\check{\phi}(p)=(2\pi)^{-2}\int \exp(ipx)\phi(x)dx$ denotes the inverse Fourier transform. Thus, since obviously $-\Delta \mathcal{L}=-2\pi\delta_0$,  the distribution 
$$\mathcal{M}(\phi):= \mathcal{L}(\phi)-\mathcal{P}(\check{\phi})$$
satisfies $-\Delta \mathcal{M}=0 $ which implies $p^2 \widehat{\mathcal{M}}=0$. Thus
$$\widehat{\mathcal{M}}=c_0 \delta_0 +c_1\partial_1\delta_0+c_2\partial_2\delta_0+c_3\partial_1\partial_2 \delta_0.$$
To determine the constants $c_i$ we will use the scaling properties of the Fourier transform. Let $\phi_\kappa (p):= \phi(\kappa p)$. Then from the definition of $\mathcal{M}$ we get
\begin{equation*}
\widehat{\mathcal{M}}(\phi_\kappa)=\mathcal{M}(\hat{\phi}_\kappa)=\mathcal{L}(\hat{\phi}_\kappa)-\mathcal{P}(\phi_\kappa).
\end{equation*}
After a change of variables we have 
\begin{equation*}
\begin{aligned}
\mathcal{P}(\phi_\kappa)=&-2\pi\left(\int_{|p|\leq \kappa}\frac{\phi(p)-\phi(0)}{p^2}dp+\int_{|p|>\kappa}\frac{\phi(p)}{p^2}dp\right) \\
&=-2\pi\left(\int_{|p|\leq 1}\frac{\phi(p)-\phi(0)}{p^2}dp+\int_{1<|p|\leq \kappa}\frac{\phi(p)-\phi(0)}{p^2}dp+\int_{|p|>\kappa}\frac{\phi(p)}{p^2}dp\right)\\
&=\mathcal{P}(\phi)+2\pi\int_{1<|p|\leq \kappa}\frac{\phi(0)}{p^2}dp=\mathcal{P}(\phi)+(2\pi)^2 \ln (|\kappa|) \phi(0)\\
&=\mathcal{P}(\phi)+\ln (|\kappa|) \int \hat{\phi}(p)dp.
\end{aligned}
\end{equation*} 
Thus
\begin{equation*}
\widehat{\mathcal{M}}(\phi_\kappa)=\mathcal{L}(\hat{\phi}_\kappa)+\mathcal{M}(\hat{\phi})-\mathcal{L}(\hat{\phi})-\ln (|\kappa|) \int \hat{\phi}(p)dp.
\end{equation*}
Since 
$$\mathcal{L}(\hat{\phi}_\kappa)=\int\ln|x|\frac{1}{\kappa^2}\hat{\phi}\left(\frac{x}{\kappa}\right)dx= \ln (|\kappa|) \int \hat{\phi}(p)dp+\mathcal{L}(\hat{\phi})$$
we arrive at
$$\widehat{\mathcal{M}}(\phi_\kappa)=\mathcal{M}(\hat{\phi}).$$
This scaling invariance together with the fact that $\delta_0 (\phi_\kappa)=\delta_0(\phi), (\partial_i\delta_0)(\phi_\kappa)=\kappa(\partial_i\delta_0)(\phi)$ and $(\partial_1\partial_2 \delta_0)(\phi_\kappa)=\kappa^2(\partial_1 \partial_2 \delta_0)(\phi)$ implies that $c_1=c_2=c_3=0$ and thus
$$\widehat{\mathcal{M}}=c_0\delta_0.$$
To compute $C_0$ we pick as test function $f(p)=\exp(-p^2/2)$. We have
$$c_0=\widehat{\mathcal{M}}(f)=\mathcal{L}(\hat{f})-\mathcal{P}(f)$$
which can be computed in polar coordinates and yields 
$$c_0=2(2\pi)^2\int_0^\infty r\ln (r) \exp(-r^2/2)dr=(2\pi)^2(\ln2 - \Gamma).$$
Thus altogether 
\bq  \label{eq:FToflog}
\widehat{\mathcal{L}}=(2\pi)^2(\ln2 - \Gamma)\delta_0+\mathcal{P}.
\eq
Remembering that away from the origin the action of the distribution $\hat{w}$ is given by \eqref{eq:hatwawayfrom 0}, plugging \eqref{eq:FToflog} into \eqref{eq:scatt_fourier_ln}, we obtain 
\begin{equation}
\begin{aligned}
\hat{\varphi}(\phi)=(2\pi)^2&(1+ 2b\ln(a)-2b\ln2 + 2b\Gamma)\phi(0) \\
&+\left(\int_{|p|\leq 1}\frac{\widehat{Vw}(p)\phi(p)-\widehat{Vw}(0)\phi(0)}{2p^2}dp+\int_{|p|>1}\frac{\widehat{Vw}(p)\phi(p)}{2p^2}dp\right).
\end{aligned}   \label{eq:FTwithoutnewdelta}
\end{equation}
Rewriting the integral $\int_{|p|\leq 1}$ as a sum $\int_{|p|\leq \epsilon} + \int_{\epsilon <|p|\leq 1}$ we can rewrite the second line in \eqref{eq:FTwithoutnewdelta} as  
\begin{equation*}
\begin{aligned}
\int_{|p|\leq 1}\frac{\widehat{Vw}(p)\phi(p)-\widehat{Vw}(0)\phi(0)}{2p^2}dp&+\int_{|p|>1}\frac{\widehat{Vw}(p)\phi(p)}{2p^2}dp  =\\
\int_{|p|\leq \epsilon}\frac{\widehat{Vw}(p)\phi(p)-\widehat{Vw}(0)\phi(0)}{2p^2}dp&+\int_{|p|>\epsilon}\frac{\widehat{Vw}(p)\phi(p)}{2p^2}dp + (2\pi)^2 2b(\ln \epsilon) \phi(0)
\end{aligned}   
\end{equation*}
where we used \eqref{eq:hatVphi0}. Choosing $\epsilon$ as in the statement of the lemma we cancel all $\delta$-terms in \eqref{eq:FTwithoutnewdelta} and obtain the desired result.
\end{proof}

\textbf{Acknowledgments.} SF was partially supported by a Sapere Aude grant from the Independent Research Fund Denmark, Grant number DFF--4181-00221. MN was supported by the National Science Centre (NCN) under the project Nr. 2016/21/D/ST1/02430. RR gratefully acknowledges the support of the Royal Society and Darwin College, Cambridge. JPS was partially supported by the Villum Centre of Excellence for the Mathematics of Quantum Theory (QMATH) and the ERC Advanced grant 321029.

\bibliographystyle{siam}

\end{document}